\documentclass[submission,copyright,creativecommons]{eptcs}
\providecommand{\event}{GASCOM 2026} 

\usepackage{iftex}

\ifpdf
  \usepackage{underscore}         
  \usepackage[T1]{fontenc}        
\else
  \usepackage{breakurl}           
\fi

\usepackage{amssymb}
\usepackage{amsthm}
\usepackage{MnSymbol}
\usepackage{stmaryrd}
\usepackage{enumitem}
\usepackage{multicol}
\usepackage{algorithm}
\usepackage{algpseudocodex}

\title{Snake Polyominoes of Maximal Area in a Rectangle}
\author{Alexandre Blondin Mass\'e
\institute{Universit\'e du Qu\'ebec \`a Montr\'eal\\
Montr\'eal, Canada}
\email{blondin\_masse.alexandre@uqam.ca}
\and
Alain Goupil
\institute{Universit\'e du Qu\'ebec \`a Trois-Rivi\`eres\\
Trois-Rivi\`eres, Canada}
\email{alain.goupil@uqtr.ca}
}
\def\titlerunning{Snake Polyominoes of Maximal Area in a Rectangle}
\def\authorrunning{A. Blondin Mass\'e \& A. Goupil}

\newtheorem{proposition}{Proposition}
\newtheorem{theorem}{Theorem}
\newtheorem{lemma}{Lemma}

\usepackage{todonotes}

\newcommand{\E}{\square}
\renewcommand{\O}{\blacksquare}
\newcommand{\AlC}{A}
\newcommand{\DI}[1]{\llbracket #1 \rrbracket}
\newcommand{\Whw}{\mathbf{W}^{\AlC}_{h \times w}}
\newcommand{\Wh}{\mathbf{W}^{\AlC}_{h \times \_}}
\newcommand{\Ww}{\mathbf{W}^{\AlC}_{\_ \times w}}

\newcommand{\Whwv}[2]{\mathbf{W}^{\AlC}_{#1 \times #2}}
\newcommand{\Wshwv}[2]{\mathbf{W}^{\mathrm{snake}}_{#1 \times #2}}
\newcommand{\Wshw}{\Wshwv{h}{w}}

\newcommand{\SsfhwHv}[3]{\mathbf{S}^{\mathrm{s-f}}_{#1,#2,#3}}
\newcommand{\SsfhwH}{\SsfhwHv{h}{w}{H}}
\newcommand{\SsfhwaHv}[4]{\mathbf{S}^{\mathrm{s-f}}_{#1,#2,#3,#4}}
\newcommand{\SsfhwaH}{\SsfhwaHv{h}{w}{a}{H}}

\newcommand{\Wsfhwv}[2]{\mathbf{W}^{\mathrm{s-f}}_{#1 \times #2}}
\newcommand{\Wsfhw}{\Wsfhwv{h}{w}}

\newcommand{\ewh}{\varepsilon_{h \times 0}}
\newcommand{\eww}{\varepsilon_{0 \times w}}
\newcommand{\area}[1]{|#1|_{\O}}

\newcommand{\N}{\mathbb{N}}
\newcommand{\Q}{\mathbb{Q}}
\newcommand{\n}{\texttt{n}}
\newcommand{\w}{\texttt{w}}
\newcommand{\e}{\texttt{e}}
\newcommand{\s}{\texttt{s}}
\newcommand{\hcat}{\overt}
\newcommand{\vcat}{\ominus}
\newcommand{\nl}{\ell}
\newcommand{\nlout}{\ell^{\mathrm{out}}}
\newcommand{\indic}[1]{\mathbb{I}\left(#1\right)}

\newcommand{\amax}{a_{\mathrm{max}}}
\newcommand{\ahmax}{\hat{a}_{\mathrm{max}}}

\newcommand\sat{\deg=2}

\begin{document}

\maketitle


\begin{abstract}
  Given a discrete rectangle $R$ of dimensions $h \times w$, let $\Wshw$ be the set of snake-like polyominoes contained in $R$ represented as binary matrices, \emph{i.e.} polyominoes whose underlying simple graph is a chain with respect to the 4-adjacency relation.
  We present an algorithm that generates $\Wshw$ for any $h$ and $w$.
  Also, let $\amax(h, w)$ the maximal area that can be realized by an element of $\Wshw$.
  We provide exact formulas of $\amax(h, w)$ for $h \in \{1,2,3,4,5\}$ and $w \in \N^*$.
\end{abstract}


\section{Introduction}

Following a previous result (\cite{blondin2025maximal}) stating that the  number of cells of maximal degree $2$ in  a 2D word with  a rectangle $R$ of dimensions $h \times w$ is  at most $2/3$ the area of $R$ plus a small constant that depends on $h$ and $w$, we now investigate the maximal area of snake polyominoes contained in a rectangle of dimensions $h \times w$.
We show  that this family of polyominoes possess  sufficient  structure  to provide  explicit expressions for their maximal area in terms of $h$ and $w$, for height $h \in \{1, 2, 3, 4, 5\}$.
Our underlying conjecture is that there exists an explicit expression for the maximal area of snake-like polyominoes in any given rectangle of arbitrary dimensions $h\times w$ that can be expressed in terms of $h$ and $w$. 
Due to space restriction, many proofs of the stated propositions and lemmas have been ommitted, but will be included in an upcoming extended version.


\section{2D Words}

For $a,b \in \N$, let $\DI{a,b} = \{n \in \N \mid a \leq n \leq b\}$.
Also, if $a,b,m \in \N$, with $m > 0$, we write $a \equiv_m b$ when $a$ and $b$ are congruent modulo $m$.
The definitions and notations that are used  here for $2$-dimensional words is adapted from~\cite{giammarresi1997two,morita2004two}.

Let $A$ be a finite alphabet and $h,w \in \N$.
A \emph{$2$-dimensional  word, or $2D$-word, $W$ of dimensions $h \times w$ on $A$} is a matrix of $h$ rows and $w$ columns with entries in $A$.
The dimensions $h$ and $w$ are called respectively the \emph{height} and \emph{width} of $W$.
The entry of $W$ at row $i \in \DI{1,h}$ and column $j \in \DI{1,w}$ is denoted by $W[i,j]$.
Also, given $i \in \DI{1,h}$ (resp. $j \in \DI{1,w}$), the $1$-dimensional word, $1D$-word, obtained by taking the $i$-th row (resp. $j$-th column) of $W$ is denoted by $W[i,\_]$ (resp. $W[\_\;,j]$).
The 2D word $W$ is called \emph{empty} if $h = 0$ or $w = 0$.
For each $h \in \N$ (resp. $w \in \N$), there exists a unique empty word of height $h$ (resp. width $w$), denoted by $\ewh$ (resp. $\eww$).
The set of all 2D words of dimensions $h \times w$ (resp. of height $h$, of width $w$)
on $A$ is denoted by $\Whw$ (resp. $\Wh$, $\Ww$).
Given $h,w_1,w_2 \in \N$, the \emph{horizontal concatenation} of $U \in \Whwv{h}{w_1}$ and $V \in \Whwv{h}{w_2}$ is the word $U \hcat V \in \Whwv{h}{(w_1+w_2)}$ defined by
	\[(U \hcat V)[i,j] = \begin{cases}
		U[i,j], & \mbox{if $j \leq w_1;$} \\
		V[i,j - w_1], & \mbox{if $j > w_1.$}
	\end{cases}\]
The \emph{vertical concatenation} $U\vcat V\in \Whwv{(h_1+h_2)}{w}$ of $U \in \Whwv{h_1}{w}$ and $V \in \Whwv{h_2}{w}$ is defined similarly for $h_1,h_2,w \in \N$.
Given $W \in \Whwv{h}{w}$ and $m,n \in \Q$ such that $mh, nw \in \N$, the \emph{$(m,n)$-th power of $W$} is the word $W^{m \times n} \in \Whwv{mh}{nw}$ such that $W^{m \times n}[i,j] = W[(i - 1\bmod h) + 1, (j - 1\bmod w) + 1]$.
One might easily prove that $(\Wh, \hcat, \ewh)$ and $(\Ww, \vcat, \eww)$ are monoids.
We denote by $\tau(W)$ the 2D word obtained from $W$ by transposing its underlying matrix.

Since 2D words can be seen as discrete rectangles filled with letters, we encode for convenience their \emph{north}, \emph{west}, \emph{east} and \emph{south} sides, or directions, with the alphabet $D = \{\n, \w, \e, \s\}$.
Also, for sake of compactness, we write subsets of $D$ as words, without, braces and commas, so that the subset $\{\n,\e,\s\}$ is denoted by $\n\e\s$.
For convenience, we consider two bijective endofunctions on $D$.
The first one is defined by $\rho(\n) = \w$, $\rho(\w) = \s$, $\rho(\e) = \n$ and $\rho(\s) = \e$, corresponding geometrically with a quarter-turn counterclockwise rotation.
The second one is defined by $\tau(\n) = \w$, $\tau(\w) = \n$, $\tau(\e) = \s$ and $\tau(\s) = \e$, corresponding geometrically with a reflection along a diagonal  axis from the north-west corner to the south-east corner.
Let $W, U$ be 2D words. $U$ is a \emph{factor} of $W$ when the rectangle of $U$ is a subrectangle  of the rectangle of $W$. More formally,
we say that $U$ is a \emph{factor} of $W$ if there exist 2D words $R_1$, $R_2$, $R_3$, $C_1$, $C_2$, $C_3$, $W_{11}$, $W_{12}$, $W_{13}$, $W_{21}$, $W_{23}$, $W_{31}$, $W_{32}$ and $W_{33}$ such that
\begin{equation}\label{eq:factor}
  \begin{array}{rcccl}
      &   & R_1   \\
      &   & \vcat \\
    W & = & R_2   & = & C_1 \hcat C_2 \hcat C_3, \\
      &   & \vcat \\
      &   & R_3   \\[3mm]
  \end{array}
  \quad
  \begin{array}{rcl}
    R_1 & = & W_{11} \hcat W_{12} \hcat W_{13} \\
    R_2 & = & W_{21} \hcat U \hcat W_{23} \\
    R_3 & = & W_{31} \hcat W_{32} \hcat W_{33} \\[3mm]
    C_1 & = & W_{11} \vcat W_{21} \vcat W_{31} \\
    C_2 & = & W_{12} \vcat U \vcat W_{32} \\
    C_3 & = & W_{13} \vcat W_{23} \vcat W_{33}
  \end{array}
\end{equation}
Given any $H \subseteq D$ and any $U, W$ such that $U$ is a factor of $W$,
we say that \emph{$U$ is an $H$-factor of $W$} if the condition $\n \in H$ implies $R_1$ is empty, the condition $\w \in H$ implies $C_1$ is empty, the condition $\e \in H$ implies $C_3$ is empty and the condition $\s \in H$ implies $R_3$ is empty.
Roughly speaking,  $U$ is an $H$-factor of $W$ if, for each $d \in H$, $U$ is on the border of $W$ in the direction $d$.

From now on, we fix the alphabet $\AlC = \{\E, \O\}$.
Given $W \in \Whw$, the \emph{graph of $W$}, denoted by $G[W]$, is the subgraph of the grid graph $G_{h \times w}$ induced by $\{(i,j) \mid W[i,j] = \O\}$.
Any ordered pair $(i,j) \in \DI{1,h} \times \DI{1,w}$ is called a \emph{cell of $W$}.
A cell $(i,j)$ is on the \emph{$\n$-boundary (resp. $\w$-boundary, $\e$-boundary, $\s$-boundary)} of $W$ if $i = 1$ (resp. $j = 1$, $j = w$, $i = h$), or simply on the \emph{boundary} of $W$ if it is on the $d$-boundary of $W$, for some $d \in D$.
Otherwise $(i,j)$ is called an \emph{internal} cell.
We say that the cell $(i,j)$ of $W$ is \emph{occupied} if $W[i,j] = \O$ and that it is \emph{empty} if $W[i,j] = \E$. The number of occupied cells in $W$ is denoted $|W|_{\O}$ and is also called the \emph{area} of $W$.
The \emph{neighborhood of $(i,j)$ in $W$} is defined by
\begin{equation}
  N_W(i,j) = \{(i',j') \in \DI{1,h} \times \DI{1,w} \mid (i - i')^2 + (j - j')^2 = 1\},
\end{equation}
while the \emph{degree} of $(i,j)$ in $W$ is $\deg_W(i,j) = \mathrm{Card}\{(i',j') \in N_W(i,j)~:~W[i',j'] = \O\}$.
A connected component $C \subseteq \DI{1,h} \times \DI{1,w}$ of $G[W]$
is called a \emph{component of $W$}.
The number of components of $W$ is denoted by $c(W)$.
We say that $W$ is \emph{$\n$-inscribed (resp. $\w$-inscribed, $\e$-inscribed, $\s$-inscribed)} if $\O$ occurs in $W[1,\_]$ (resp. $W[\_,1]$, $W[\_,w]$, $W[h,\_]$).
Finally, we say that $W$ is \emph{inscribed} if it is $d$-inscribed for all $d \in D$.


\section{Snake Words and Snake Factors}

A \emph{snake word} is an inscribed binary word $W$ on $\AlC$ such that $G[W]$ is a chain.
The set of all snake words of dimensions $h \times w$ on $\AlC$ is denoted by $\Wshw$.
Similarly, a \emph{snake forest word} is a binary word $W$ on $\AlC$ such that $G[W]$ is a forest of chains.
A binary word $W$ is called a \emph{snake factor} if there exists a snake word $W'$ such that $W$ is a factor of $W'$.
The set of all snake factors of dimensions $h \times w$ on $\AlC$ is denoted by $\Wsfhw$.
The following two observations are immediate:
\begin{proposition} ~
  \begin{enumerate}[label=(\roman*),noitemsep]
    \item A snake factor is a snake forest word.
    \item There exist snake forest words that are not snake factors.
  \end{enumerate}
\end{proposition}

%

A \emph{sided word of dimensions $h \times w$} is a pair $S = (W, H)$, where $W \in \Whw$ and $H \subseteq D$ is called the \emph{hull} of $S$.
Without ambiguity, the  area and  number of components are naturally extended from words to sided words by setting $\area{S} = \area{W}$ and $c(S) = c(W)$.
Similarly, the definitions of \emph{height}, \emph{width}, \emph{cell}, \emph{degree of a cell}, \emph{component}, \emph{snake word}, \emph{snake forest} and \emph{inscription} on 2D words are naturally extended to sided words.

The definitions of \emph{horizontal concatenation}, \emph{vertical concatenation}, \emph{factor} and \emph{snake factor} for sided words, however, need to be adjusted.
Let $S = (W, H)$ and $S' = (W', H')$ be two sided words.
Then the \emph{horizontal concatenation} of $S$ and $S'$ is the partial binary operations $S \hcat S' = (W \hcat W', H \cup H')$ defined whenever $S$ and $S'$ have the same height, $\e \notin H$, $\w \notin H'$ and $H - \{\w\} = H' - \{\e\}$.
The \emph{vertical concatenation} $S \vcat S'$ of $S$ and $S'$ is defined similarly.
We say that $S$ is a \emph{factor} of $S'$ if $H \subseteq H'$ and $W$ is an $H$-factor of $W'$.
Finally, $S = (W, H)$ is called a \emph{snake factor} if there exists an inscribed sided word $S' = (W', D)$ such that $S$ is a factor of $S'$ and $W$ is a snake factor of $W'$.

Let $S = (W, H)$ be a sided word and $(i,j)$ a cell of $S$.
The \emph{number of outside liberties of $(i,j)$ in $S$}, denoted by $\nlout_S(i,j)$, is defined by
\begin{equation}\begin{array}{rcl}
  \nlout_S(i,j)
    & = & \indic{i = 1}\indic{\n \notin H} + \indic{j = 1}\indic{\w \notin H} \\
    & + & \indic{i = h}\indic{\s \notin H} + \indic{j = w}\indic{\e \notin H},
\end{array}\end{equation}
where $\indic{\cdot}$ is the usual indicator function, \emph{i.e.} $\indic{p} = 1$ if $p$ is true, $0$ otherwise.
Similarly, the \emph{number of liberties of $(i,j)$ in $S$}, denoted by $\nl_S(i,j)$, is defined by
\begin{equation}
  \nl_S(i,j) = \max(\min(\nlout_S(i,j), 2 - \deg_S(i, j)), 0).
\end{equation}
It follows from the definition that $\nl_S(i,j) \in \{0, 1, 2\}$ if $(i,j)$ is a boundary cell, while $(i,j)$ is an internal cell implies that $\nl_S(i,j) = 0$. 
Intuitively, the number of liberties of a cell is the maximum number of its edges that can be connected to other cells without creating cells of degree greater than $2$.
Let $C$ be a component of $S$ and define the number of liberties of $C$ as $\nl_S(C) = \sum_{(i,j) \in C} \nl_S(i,j)$.
Then $C$ is called \emph{unconnectable (resp. terminal, crossing)} if $\nl_S(C) = 0$ (resp. $\nl_S(C) = 1$, $\nl_S(C) = 2$).
Moreover, given $d \in D$, we say that \emph{$C$ is connectable on $d$ in $S$} if there exists a $d$-boundary cell $(i,j) \in C$ such that $\nl_S(i,j) > 0$.
The number of unconnectable, crossing, terminal components of $S$ are respectively  denoted  $u(S)$. $x(S)$, $t(S)$. For each $d \in D$, the number of terminal components connectable on $d$ is denoted by $t_d(S)$.
The first main result of this extended abstract is a characterization of sided snake factors:
\begin{theorem}\label{thm:snake-factor}
  Let $S = (W, H)$ be a sided word.
  Then $S$ is a sided snake factor if and only if $S$ is a sided snake forest word, $t(S) \leq 2$, and exactly one of the following conditions holds:
  \begin{enumerate}[label=(\roman*),noitemsep]
    \item $H = D$, $S$ is inscribed and $c(S) = 1$.
    \item $H = D - \{d\}$ for some $d\in D$, $S$ is $\rho^2(d)$-inscribed and $u(S) = 0$.
    \item $|H| \in \{0, 1\}$ or $H = \{d, \rho(d)\}$ for some $d$, and $u(S) = 0$.
    \item $H = D - \{d, \rho^2(d)\}$ for some $d$, $u(S) = 0$, $S$ has at least one crossing component connectable on both $d$ and $\rho^2(d)$, and the condition $t(S) = 2$ implies $t_d(S), t_{\rho^2(d)}(S) \equiv_2 x(S)$.
  \end{enumerate}
\end{theorem}

In order to prove Theorem~\ref{thm:snake-factor}, we need the following lemmas:
\begin{lemma}\label{lem:snake-obs}
  Let $S = (W,H)$ be a sided snake factor and $C$ a terminal component of $S$.
  Then exactly one of the following two conditions holds:
  \begin{enumerate}[label=(\arabic*), noitemsep]
    \item $S$ has a unique internal cell $(i,j) \in C$ such that $\deg_S(i,j) = 1$.
    \item $S$ has exactly one boundary cell $(i,j) \in C$ such that $\deg_S(i,j) = 0$ and $\nlout_S(i,j) = 1$.
  \end{enumerate}
  The unique cell satisfying  Condition (1) or (2) is called the \emph{terminal} cell of $C$.
\end{lemma}

The following observation is well-known in graph theory:
\begin{lemma}\label{lem:snake-chain}
  Let $S = (W,H)$ be a sided snake word and $C$ the only component of $S$.
  Then for any pair of cells $c$ and $c'$ of $C$, there exists a unique chain between $c$ and $c'$.
\end{lemma}

We are now ready to prove theorem \ref{thm:snake-factor}.

\begin{proof}[Proof of Theorem~\ref{thm:snake-factor}]
  $(\Rightarrow)$ Assume that $S = (W,H)$ is a sided snake factor.
  Then there exists an inscribed sided snake word $S' = (W',D)$
  such that $S$ is a factor of $S'$.
  Let $G = G[W]$ and $G' = G[W']$.
  First, since $G$ is a subgraph of $G'$, we conclude that $G$ is a snake forest, so that $S$ is a sided snake forest word.
  Next, we cannot have $t(S) \geq 3$: if it was the case, by Lemma \ref{lem:snake-obs}, $G'$ would have at least $3$ vertices of degree $\leq 1$, contradicting the fact that $G'$ is a chain.
  Hence, $t(S) \leq 2$.
  To conclude this part, we consider four cases.

  \textit{Case $H = D$}.
  Since $S$ is a factor of $S'$, by definition of $H$-factor, we conclude that $W = W'$, which means that $S = S'$ is a snake word.
  Hence, $c(S) = 1$.

  \textit{Case $H = D - \{d\}$}.
  Then $S$ is $\rho^2(d)$-inscribed: if it was not the case, then $S'$ would not be inscribed.
  Moreover, we cannot have $u(S) > 0$: if it was the case, then we would have $c(S') > 1$.

  \textit{Case $|H| \in \{0, 1\}$ or $H = \{d, \rho(d)\}$}.
  As for the preceding case, the assumption $u(S) > 0$ implies $c(S') > 1$, a contradiction.

  \textit{Case $H = D - \{d, \rho^2(d)\}$}.
  The fact that $u(S) = 0$ follows from an argument similar to those of the two previous cases.
  Also, $S$ must have at least one crossing component connectable on both $d$ and $\rho^2(d)$, otherwise, $S'$ would have at least two components, on each open sides of $S$.
  Finally, assume that $t(S) = 2$ and let $t = (t_d(S), t_{\rho^2(d)}(S))$, which implies $t \in \{(0, 2), (1, 1), (2, 0)\}$.
  If $t = (1, 1)$, then $x(S)$ must be odd.
  Indeed, by Lemma~\ref{lem:snake-chain}, there exists a unique chain between the two terminal cells of the only component of $S'$.
  Since the terminal components of $S$ are connected on opposite sides of $S$
  , they must cross $S$ an odd number of times.
  Using a similar argument, if $t \in \{(0, 2), (2, 0)\}$, then $x(S)$ must be even.
  Since all cases have been covered, the result follows.

  $(\Leftarrow)$ Due to space restriction, we only provide an idea of the proof.
  We need to prove that there exists an inscribed sided word $S' = (W',D)$ such that $S$ is a factor of $S'$ and $W$ is a snake factor of $W'$.
  There are four cases to consider, according to which condition between (i) and (iv) holds.

  \textit{Case (i)}. If suffices to take $S' = S$.

  \textit{Case (ii)}. If suffices to connect all components of $S$ by extending it on its only open side, by making sure that the result is $d$-inscribed, for all $d \in D - \{\rho^2(d)\}$.

  \textit{Case (iii)}. Similarly to case (ii), il suffices to connect all components of $S$ and by making sure that the result is $d$-inscribed in all directions.

  \textit{Case (iv)}. In that case, the inscription is easily obtained in all directions.
  The connectivity is guaranteed by the conditions on the components: the existence of at least one crossing component allows to connect each side, and if there are two terminal components, then they can be connected consistently with the crossing components, thanks to their equivalence modulo $2$.
\end{proof}


\section{Generation of Sided Snake Factors}

Let $\SsfhwaH$ be the set of all sided snake factors of dimensions $h \times w$, of area $a$ and of hull $H$.
We were able to use Algorithm~\ref{algo:bsf-generation} to generate $\SsfhwaH$ for small dimensions $h \times w$.

\begin{algorithm}[t]
  \caption{Generation of sided snake factors}\label{algo:bsf-generation}
  \begin{algorithmic}
    \Function{BSF}{$h, w, a$ : positive integers, $H$ : hull}
      \If{$h > w$}
        \State \Return $\tau(\Call{BSF}{w, h, a, \tau(H)})$
      \ElsIf{$(h, w) = (1, 1)$}
        \If{$a = 1$}
          \State \Return $\{(\O,H)\}$
        \ElsIf{$a = 0$ and $|H| \neq 3$ and $H \notin \{\n\s,\w\e\}$}
          \State \Return $\{(\E,H)\}$
        \Else
          \State \Return $\emptyset$
        \EndIf
      \Else
        \State $(w_1, w_2) \gets (\lfloor w / 2 \rfloor, \lceil w / 2 \rceil)$
        \State $(H_1,H_2) \gets (H - \{\e\}, H - \{\w\})$
        \State $S \gets \emptyset$
        \For{$a_1 \in \{0,1,\ldots,a\}$}
          \State $a_2 \gets a - a_1$
          \State $L \gets \Call{BSF}{h, w_1, a_1, H_1}$
          \State $R \gets \Call{BSF}{h, w_2, a_2, H_2}$
          \State $S \gets S \cup \{\ell \hcat r \mid \ell \in L, r \in R \mbox{ and $\ell \hcat r$ is a snake factor}\}$
        \EndFor
        \State \Return $S$
      \EndIf
    \EndFunction
  \end{algorithmic}
\end{algorithm}

We need additional definitions and notation.
First, let $S = (W, H) \in \SsfhwH$.
We denote by $|S|_{\E}$ the $4$-tuple $(n, w, e, s)$ indicating, for each side $d$, the number of empty cells of $S$ located on the side $d$ boundary of $S$ if $d \in H$, and we write $\_$ otherwise.
Similarly, we denote by $|S|_{\sat}$ the $4$-tuple $(n, w, e, s)$ indicating, for each side $d$, the number of cells of $S$, whose degree in $S$ is $2$, located on the side $d$ boundary of $S$ if $d \in H$, and we write $\_$, otherwise.
For any $h, w \in \N^*$ and $H \subseteq D$, let $\SsfhwH$ be the set of all sided snake factors of dimensions $h \times w$ and of hull $H$, whatever their area, and let $\amax(h, w, H)$ be the maximal area that can be realized by a sided snake factor of dimensions $h \times w$ with hull $H$, i.e. $\amax(h,w,H) = \max\{\area{S} : S \in \SsfhwH\}$.
The following observation is immediate:

\begin{lemma}\label{lem:amax-convexity}
  Let $h,w,w_1,w_2 \in \N^*$, with $w = w_1 + w_2$, and $H \subseteq D$.
  Then
  \begin{equation}
    \amax(h,w,H) \leq \amax(h,w_1,H - \{\e\}) + \amax(h,w_2,H - \{\w\}).
  \end{equation}
\end{lemma}

Algorithm~\ref{algo:bsf-generation} can be helpful in establishing exact values or upper bounds of $\amax(h, w, H)$.
More precisely:

\begin{lemma}\label{lem:amax}
  Let $h,w,a \in \N^*$ and $H \subseteq D$, where $h + w - 1 \leq a \leq hw$.
  If $\SsfhwaHv{h}{w}{a}{H} \neq \emptyset$ and $\SsfhwaHv{h}{w}{a + 1}{H} = \emptyset$, then $\amax(h, w, H) = a$.
\end{lemma}

As a consequence, we obtain the following observations for $h = 1, 2, 3$.

\begin{lemma}\label{lem:amax-123}
  Let $w \in \N^*$ and $H \subseteq D$.
  Then the following identities hold:
  \begin{enumerate}[label=(\roman*),noitemsep]
    \item $\amax(1, 1, H) = 1$;
    \item $\amax(2, 1, H) = 2$,
          $\amax(2, 2, H) = 3$;
    \item $\amax(2, w, \n\s) \leq 3w/2$ if $w \bmod 2 = 0$;
    \item $\amax(3, 1, H) = 3$,
          $\amax(3, 2, \n\w\s) = \amax(3, 2, \n\e\s) = 5$,
          $\amax(3, 2, \n\s) = 4$;
    \item $\amax(3, w, \n\s) \leq 4w/2$ if $w \bmod 2 = 0$;
  \end{enumerate}
\end{lemma}

\begin{proof}
  (i), (ii) and (iv) follow from direct observation,  Algorithm~\ref{algo:bsf-generation} and  Lemma~\ref{lem:amax}.

  (iii) follows from (ii),  Lemma~\ref{lem:amax-convexity} and from induction on $w / 2$.

  (v) follows from (iv),  Lemma~\ref{lem:amax-convexity} and from induction on $w / 2$.
\end{proof}

The cases $h = 4$ and $h = 5$ show more complex periodicity patterns:

\begin{lemma}\label{lem:amax-4}
  %
  For any $w \in \N^*$,
  \[
    \amax(4, w, \n\w\s) \leq \begin{cases}
      (8w + 4)/3, & \mbox{if $w = 1$;} \\
      (8w + 3)/3, & \mbox{if $w \bmod 3 = 0$;} \\
      (8w + 1)/3, & \mbox{if $w \bmod 3 = 1$;} \\
      (8w + 2)/3, & \mbox{if $w \bmod 3 = 2$.}
    \end{cases}
    \enskip
    \amax(5, w, \n\w\s) \leq \begin{cases}
      (22w + 13)/7, & \mbox{if $w \in \{1,8\}$;} \\
      (22w + 15)/7, & \mbox{if $w = 6$;} \\
      (22w + 14)/7, & \mbox{if $w = 7$;} \\
      (22w + 7)/7, & \mbox{if $w \bmod 7 = 0$;} \\
      (22w + 6)/7, & \mbox{if $w \bmod 7 = 1$;} \\
      (22w + 12)/7, & \mbox{if $w \bmod 7 = 2$;} \\
      (22w + 11)/7, & \mbox{if $w \bmod 7 = 3$;} \\
      (22w + 10)/7, & \mbox{if $w \bmod 7 = 4$;} \\
      (22w + 9)/7, & \mbox{if $w \bmod 7 = 5$;} \\
      (22w + 8)/7, & \mbox{if $w \bmod 7 = 6$.}
    \end{cases}
  \]
\end{lemma}



\section{Main Result}

Let $\amax(h, w)$ the the maximal area that can be realized by a snake-like polyomino inscribed in a rectangle $h \times w$.
Also, for $w \in \N^*$, let
\begin{multicols}{2}
  \vspace{0pt}\strut
  \begin{align*}
    \ahmax(1, w) & = w \\
    \ahmax(2, w)
      & = \begin{cases}
        3w / 2, & \mbox{if $w \bmod 2 = 0$;} \\
        (3w + 1) / 2, & \mbox{if $w \bmod 2 = 1$;}
      \end{cases} \\
    \ahmax(3, w)
      & = \begin{cases}
        2w + 1, & \mbox{if $1 \leq w \leq 5$;} \\
        2w + 2, & \mbox{if $w \geq 6$.}
      \end{cases} \\
    \ahmax(4, w)
      & = \begin{cases}
        (8w + 1)/3, & \mbox{if $w = 4$;} \\
        (8w + 3)/3, & \mbox{if $w \bmod 3 = 0$;} \\
        (8w + 4)/3, & \mbox{if $w \bmod 3 = 1$;} \\
        (8w + 2)/3, & \mbox{if $w \bmod 3 = 2$.}
      \end{cases}
  \end{align*}
  \vspace{0pt}\strut
  \begin{align*}
    \ahmax(5, w)
      & = \begin{cases}
        3w + 4, & \mbox{if $w \in \{10, 11, 12\}$;} \\
        (22w + 14)/7, & \mbox{if $w \bmod 7 = 0$;} \\
        (22w + 13)/7, & \mbox{if $w \bmod 7 = 1$;} \\
        (22w + 12)/7, & \mbox{if $w \bmod 7 = 2$;} \\
        (22w + 11)/7, & \mbox{if $w \bmod 7 = 3$;} \\
        (22w + 10)/7, & \mbox{if $w \bmod 7 = 4$;} \\
        (22w + 9)/7, & \mbox{if $w \bmod 7 = 5$;} \\
        (22w + 15)/7, & \mbox{if $w \bmod 7 = 6$.}
      \end{cases}
  \end{align*}
\end{multicols}

The main result of this extended abstract is the following.

\begin{theorem}\label{thm:main}
  For $1 \leq h \leq 5$ and $w \in \N^*$, $\amax(h, w) = \ahmax(h, w)$.
\end{theorem}

The proof of Theorem~\ref{thm:main} is divided in two cases, that are addressed in the following two lemmas.

\begin{lemma}\label{lem:lower-bound}
  For $1 \leq h \leq 5$ and $w \in \N^*$, $\amax(h, w) \geq \ahmax(h, w)$.
\end{lemma}

\begin{proof}
  It suffices to exhibit families of snakes having area $\ahmax(h, w)$ for $1 \leq h \leq 5$ and for any $w \in \N^*$.
  Let
  \[
    S_{1,w} = \O^{1 \times w}, \quad
    S_{2,w} = \left(\begin{array}{@{}c@{}c@{}c@{}c@{}}
      \O & \O & \O & \E \\[-2mm]
      \O & \E & \O & \O
    \end{array}\right)^{1 \times (w / 4)} \quad \text{and} \quad
    S_{3,w} = \begin{cases}
      \left(\begin{array}{@{}c@{}c@{}c@{}c@{}c@{}c@{}}
        \O & \O & \O & \O & \E & \\[-2.5mm]
        \O & \E & \E & \O & \O & \\[-2.5mm]
        \O & \O & \O & \E & \O &
      \end{array}\right)^{1 \times (w / 5)}, & \mbox{if $1 \leq w \leq 5$;} \\
      \begin{array}{@{}c@{}c@{}c@{}c@{}c@{}}
        \O & \O & \O & \E & \O \\[-2.5mm]
        \O & \E & \O & \O & \E \\[-2.5mm]
        \O & \O & \E & \O & \O
      \end{array}
      \left(\begin{array}{@{}c@{}c@{}}
        \O \\[-2.5mm]
        \E \\[-2.5mm]
        \O
      \end{array}\right)^{1 \times (w - 6)}
      \begin{array}{@{}c@{}c@{}}
        \O \\[-2.5mm]
        \O \\[-2.5mm]
        \O
      \end{array}, & \mbox{if $w \geq 6$;}
    \end{cases}
  \]
  Moreover, let
  \begin{align}
    S_{4,w} & = \begin{cases}
      A_4^{1 \times (w/2)},
        & \mbox{if $1 \leq w \leq 2$;} \\
      A_4 (B_4 B_4^h)^{1 \times((w-2)/6)},
        & \mbox{if $3 \leq w \leq 6$ or $w \bmod 3 \neq 1$;} \\
      A_4 (B_4 B_4^h)^{1 \times ((w-7)/6)} C_4,
        & \mbox{if $w \geq 7$ and $w \bmod 6 = 1$;} \\
      A_4 (B_4 B_4^h)^{1 \times ((w-7)/6)} B_4 C_4^h,
        & \mbox{if $w \geq 7$ and $w \bmod 6 = 4$.}
    \end{cases}
  \end{align}
  where
  \[
    A_4 = \begin{array}{@{}c@{}c@{}}
      \O & \O \\[-2mm]
      \O & \E \\[-2mm]
      \O & \E \\[-2mm]
      \O & \O
    \end{array}, \quad
    B_4 = \begin{array}{@{}c@{}c@{}c@{}}
      \O & \O & \O \\[-2mm]
      \E & \E & \E \\[-2mm]
      \O & \O & \O \\[-2mm]
      \O & \E & \O
    \end{array}, \quad
    B_4^h = \begin{array}{@{}c@{}c@{}c@{}}
      \O & \E & \O \\[-2mm]
      \O & \O & \O \\[-2mm]
      \E & \E & \E \\[-2mm]
      \O & \O & \O
    \end{array}, \quad
    C_4 = \begin{array}{@{}c@{}c@{}c@{}c@{}c@{}}
      \O & \E & \O & \O & \O \\[-2mm]
      \E & \O & \O & \E & \O \\[-2mm]
      \O & \O & \E & \O & \O \\[-2mm]
      \O & \E & \O & \O & \E
    \end{array} \quad \mbox{and} \quad
    C_4^h = \begin{array}{@{}c@{}c@{}c@{}c@{}c@{}}
      \O & \E & \O & \O & \E \\[-2mm]
      \O & \O & \E & \O & \O \\[-2mm]
      \E & \O & \O & \E & \O \\[-2mm]
      \O & \E & \O & \O & \O
    \end{array}.
  \]
  Finally, let
  \begin{align}
    S_{5,w} & = \begin{cases}
      A_5,
        & \mbox{if $w = 1$;} \\
      B_5 C_5^{1 \times (w-2)} B_5^h,
        & \mbox{if $2 \leq w \leq 5$;} \\
      B_5 C_5^{1 \times (w-6)} D_5^r,
        & \mbox{if $6 \leq w \leq 9$;} \\
      D_5 C_5^{1 \times (w-10)} D_5^r,
        & \mbox{if $10 \leq w \leq 19$;} \\
      E_5 C_5^{1 \times m} F_5 (G_5 G_5^h)^{1 \times n} H_5^h,
        & \mbox{if $w \geq 20$ and $0 \leq w \bmod 14 \leq 6$;} \\
      E_5 C_5^{1 \times m} F_5 (G_5 G_5^h)^{1 \times n} H_5,
        & \mbox{if $w \geq 20$ and $7 \leq w \bmod 14 \leq 13$.}
    \end{cases}
  \end{align}
  where $m = (w+1) \bmod 7 + 1$, $n = (w-m-12)/14$ and
  \[
    A_5 = \begin{array}{@{}c@{}}
      \O \\[-2mm]
      \O \\[-2mm]
      \O \\[-2mm]
      \O \\[-2mm]
      \O
    \end{array}, \enskip
    B_5 = \begin{array}{@{}c@{}}
      \O \\[-2mm]
      \E \\[-2mm]
      \O \\[-2mm]
      \O \\[-2mm]
      \O
    \end{array}, \enskip
    B_5^h = \begin{array}{@{}c@{}}
      \O \\[-2mm]
      \O \\[-2mm]
      \O \\[-2mm]
      \E \\[-2mm]
      \O
    \end{array}, \enskip
    C_5 = \begin{array}{@{}c@{}}
      \O \\[-2mm]
      \E \\[-2mm]
      \O \\[-2mm]
      \E \\[-2mm]
      \O
    \end{array}, \enskip
    D_5 = \begin{array}{@{}c@{}c@{}c@{}c@{}c@{}}
      \O & \O & \E & \O & \O \\[-2mm]
      \O & \E & \O & \O & \E \\[-2mm]
      \O & \E & \O & \E & \O \\[-2mm]
      \O & \E & \O & \E & \O \\[-2mm]
      \O & \O & \O & \E & \O
    \end{array}, \enskip
    D_5^r = \begin{array}{@{}c@{}c@{}c@{}c@{}c@{}}
      \O & \E & \O & \O & \O \\[-2mm]
      \O & \E & \O & \E & \O \\[-2mm]
      \O & \E & \O & \E & \O \\[-2mm]
      \E & \O & \O & \E & \O \\[-2mm]
      \O & \O & \E & \O & \O
    \end{array}, \enskip
    E_5 = \begin{array}{@{}c@{}c@{}c@{}c@{}c@{}}
      \O & \O & \O & \E & \O \\[-2mm]
      \O & \E & \O & \O & \O \\[-2mm]
      \O & \E & \E & \E & \E \\[-2mm]
      \O & \E & \O & \O & \O \\[-2mm]
      \O & \O & \O & \E & \O
    \end{array},
  \]
  \[
    F_5 = \begin{array}{@{}c@{}c@{}}
      \E & \O \\[-2mm]
      \O & \O \\[-2mm]
      \O & \E \\[-2mm]
      \E & \E \\[-2mm]
      \O & \O
    \end{array}, \enskip
    G_5 = \begin{array}{@{}c@{}c@{}c@{}c@{}c@{}c@{}c@{}}
      \O & \O & \O & \O & \E & \O & \O \\[-2mm]
      \E & \E & \E & \O & \O & \O & \E \\[-2mm]
      \O & \O & \O & \E & \E & \E & \E \\[-2mm]
      \O & \E & \O & \E & \O & \O & \O \\[-2mm]
      \O & \E & \O & \O & \O & \E & \O
    \end{array}, \enskip
    G_5^h = \begin{array}{@{}c@{}c@{}c@{}c@{}c@{}c@{}c@{}}
      \O & \E & \O & \O & \O & \E & \O \\[-2mm]
      \O & \E & \O & \E & \O & \O & \O \\[-2mm]
      \O & \O & \O & \E & \E & \E & \E \\[-2mm]
      \E & \E & \E & \O & \O & \O & \E \\[-2mm]
      \O & \O & \O & \O & \E & \O & \O
    \end{array}, \enskip
    H_5 = \begin{array}{@{}c@{}c@{}c@{}c@{}c@{}}
      \O & \O & \O & \O & \O \\[-2mm]
      \E & \E & \E & \E & \O \\[-2mm]
      \O & \O & \O & \E & \O \\[-2mm]
      \O & \E & \O & \E & \O \\[-2mm]
      \O & \E & \O & \O & \O
    \end{array}, \enskip
    H_5^h = \begin{array}{@{}c@{}c@{}c@{}c@{}c@{}}
      \O & \E & \O & \O & \O \\[-2mm]
      \O & \E & \O & \E & \O \\[-2mm]
      \O & \O & \O & \E & \O \\[-2mm]
      \E & \E & \E & \E & \O \\[-2mm]
      \O & \O & \O & \O & \O
    \end{array}.
  \]
  Then, for any $h = 1, 2, 3, 4, 5$ and $w \in \N^*$, the dimensions of $S_{h,w}$ are $h \times w$ and $\area{S_{h,w}} = \ahmax(h, w)$.
  Moreover, by using induction on $w$, we have that $S_{h,w}$ is a snake, concluding the proof.
\end{proof}
There is  one last result to prove.

\begin{lemma}\label{lem:upper-bound}
  For $1 \leq h \leq 5$ and $w \in \N$, $\amax(h, w) \leq \ahmax(h, w)$.
\end{lemma}

\begin{proof}
  Let $S \in \SsfhwHv{h}{w}{D}$ with $\area{S} = \amax(h,w)$, i.e. $S$ is a sided snake of maximal area inscribed in a rectangle of dimensions $h \times w$.
  We show that $\area{S} \leq \ahmax(h, w)$.

  Case $h = 1$.
  It suffices to notice that $\area{S} \leq w = \ahmax(1,w)$.

  Case $h = 2$.
  The cases $w = 1$ and $w = 2$ can be proved by exhaustive verification.
  Now, assume that $w \geq 3$.
  Let $w_3 = 1 + (w - 1) \bmod 2$ and $w_2 = w - w_3 - 2$.
  Notice that $w_2 \bmod 2 = 0$.
  Then $S = S_1S_2S_3$, for some $S_1 \in \SsfhwHv{2}{2}{\n\w\s}$, $S_2 \in \SsfhwHv{2}{w_2}{\n\s}$ and $S_3 \in \SsfhwHv{2}{w_3}{\n\e\s}$.
  By Lemma~\ref{lem:amax-123}(ii)-(iii),
  \begin{align*}
    \area{S}
      & = \area{S_1} + \area{S_2} + \area{S_3} \\
      & \leq 3 + 3w_2/2 + (w_3 + 1) = (3w + 2 - w_3)/2 = \begin{cases}
          3w/2, & \mbox{if $w \bmod 2 = 0$;} \\
          (3w + 1)/2, & \mbox{if $w \bmod 2 = 1$.}
        \end{cases} \\
      & = \amax(2, w)
  \end{align*}

  Case $h = 3$.
  This case follows from \cite{blondin2025maximal}.
  However, for sake of consistency, we provide an alternate proof based on the tools introduced in the previous pages.
  First, the cases $1 \leq w \leq 5$ can be proved by exhaustive verification.
  Now, assume that $w \geq 6$.
  Let $w_3 = 1 + (w - 1) \bmod 2$ and $w_2 = w - w_3 - 1$.
  Notice that $w_2 \bmod 2 = 0$.
  Then $S = S_1S_2S_3$, for some $S_1 \in \SsfhwHv{3}{1}{\n\w\s}$, $S_2 \in \SsfhwHv{3}{w_2}{\n\s}$ and $S_3 \in \SsfhwHv{3}{w_3}{\n\e\s}$.
  By Lemma~\ref{lem:amax-123}(ii)-(iii),
  \begin{align*}
    \area{S}
      & = \area{S_1} + \area{S_2} + \area{S_3} \leq 3 + 4w_2/2 + (2w_3 + 1) = 2w + 2 = \amax(3, w)
  \end{align*}

  Case $h = 4$.
  The cases $1 \leq w \leq 5$ can be proved by exhaustive verification.
  Now, assume that $w \geq 6$.
  Then $S = S_1S_2$, for some $S_1 \in \SsfhwHv{4}{w - 4}{\n\w\s}$ and $S_2 \in \SsfhwHv{4}{4}{\n\e\s}$.
  By Lemma~\ref{lem:amax-4}(iv),
  \begin{align*}
    \area{S} = \area{S_1} + \area{S_2}
      & \leq \left(\begin{cases}
        (8(w - 4) + 3)/3, & \mbox{if $(w - 4) \bmod 3 = 0$;} \\
        (8(w - 4) + 1)/3, & \mbox{if $(w - 4) \bmod 3 = 1$;} \\
        (8(w - 4) + 2)/3, & \mbox{if $(w - 4) \bmod 3 = 2$.}
      \end{cases}\right)
      + 11 \\
      & = \begin{cases}
        (8w + 4)/3, & \mbox{if $w \bmod 3 = 1$;} \\
        (8w + 2)/3, & \mbox{if $w \bmod 3 = 2$;} \\
        (8w + 3)/3, & \mbox{if $w \bmod 3 = 0$.}
      \end{cases}
  \end{align*}

  Case $h = 5$.
  The structure of the proof is similar to the case $h = 4$, but requires to study more cases.
  In particular, one shows that all concatenations of maximal elements of $\SsfhwHv{5}{w}{\n\w\s}$ and $\SsfhwHv{5}{7}{\n\s}$ yield invalid concatenations for small values of $w$, and that maximal elements of $\SsfhwHv{5}{w}{\n\w\s}$ for any $w$ satisfy some boundary conditions.
\end{proof}

\begin{proof}[Proof of Theorem~\ref{thm:main}]
  Follows from Lemmas \ref{lem:lower-bound} and \ref{lem:upper-bound}.
\end{proof}


\section{Conclusion}

Computations not reported here for larger values of $h$ suggest that similar formulas might exist for any $h$.
Hence, we believe that the results presented in the previous sections introduce arguments that can be extended for the area of snake-like polyominoes in any rectangle.


\bibliographystyle{eptcs}
\bibliography{gascom2026}

\end{document}